\documentclass[12pt]{article}

\usepackage{natbib}
\usepackage{bm}
\usepackage{dsfont}
\usepackage{latexsym}
\usepackage{amsmath}
\usepackage{amsthm}
\usepackage{amssymb}
\usepackage{amsfonts}
\usepackage[noend]{algorithmic}
\usepackage{algorithm}
\usepackage{graphicx}
\usepackage{color}
\usepackage[all]{xy}
\usepackage[pdfpagelabels,linkcolor=blue,urlcolor=blue,pdfstartview=FitH,pdfpagemode=None]{hyperref}
\usepackage{enumitem}
\usepackage[margin=1in]{geometry}
\usepackage{cuted,balance}
\usepackage{framed}

\bmdefine\bmu{\mu}

\newcommand{\Reals}{\mathbb{R}}

\newcommand{\PosReals}{\Reals_+}

\newcommand{\AutoAdjust}[3]{\mathchoice{ \left #1 #2  \right #3}{#1 #2 #3}{#1 #2 #3}{#1 #2 #3} }
\newcommand{\Abs}[1]{{\AutoAdjust{\lvert}{#1}{\rvert}}}

\newcommand{\InBrackets}[1]{\AutoAdjust{[}{#1}{]}}
\newcommand{\BigO}[1]{{O\AutoAdjust{(}{{#1}}{)}}}

\newcommand{\RangeAB}[2]{\AutoAdjust{[}{{#1},{#2}}{]}}

\newcommand{\RangeN}[1]{\InBrackets{#1}}

\newcommand{\EQ}{W}
\newcommand{\EQH}{\overline{\EQ}}
\newcommand{\EQL}{\underline{\EQ}}

\newcommand{\Market}{M}
\newcommand{\Buyers}{I}
\newcommand{\Goods}{J}
\newcommand{\UFun}{{\mathrm u}}
\newcommand{\PFun}{{\mathrm p}}

\newcommand{\MaxVal}{m}
\newcommand{\UVec}{{\mathbf u}}
\newcommand{\PVec}{{\mathbf p}}
\newcommand{\MVec}{{\bmu}}
\newcommand{\Val}{v}
\newcommand{\CTR}{c}
\newcommand{\CTRS}{\hat{c}}

\newcommand{\Alloc}{{y}}
\newcommand{\ISubsetS}{S}
\newcommand{\GSubsetT}{T}
\newcommand{\DemandSet}{D}

\newcommand{\UVecLB}{{\underline \UVec}}
\newcommand{\PVecLB}{{\underline \PVec}}
\newcommand{\ZeroVec}{\mathbf{0}}

\newcommand{\BREAKCOL}{&}

\newcommand{\argmax}{\operatorname{argmax}}

\newtheorem{theorem}{Theorem}%

\newtheorem{lemma}{Lemma}%

\newtheorem{definition}{Definition}%

\newtheorem{example}{Example}%

\newtheorem{corollary}{Corollary}%

\numberwithin{equation}{section}%

\newcommand{\tosaeed}{}

\title{Competitive Equilibria in Two Sided Matching Markets with Non-transferable Utilities~\protect\footnote{$ $Id: ce.tex 315 2012-11-13 22:09:56Z saeed $ $}}

\date{}

\author{
Saeed Alaei \tosaeed \thanks {Dept.\ of Computer Science, Cornell University, Ithaca, NY 14850.
\texttt{saeed@cs.cornell.edu}. Part of this work was done when the author was visiting Microsoft Research,
Redmond. This work was partially supported by the NSF grant CCF-0728839},
Kamal Jain\thanks {eBay Research, 2065 Hamilton Ave, San Jose, 95125 \texttt{kamaljain@gmail.com}.},
Azarakhsh Malekian\thanks {Laboratory of Information and Decision Systems, MIT, Cambridge, MA 02142.
\texttt{malekian@mit.edu}.}}

\begin{document}
\maketitle


\begin{abstract}

We consider two sided matching markets consisting of agents with non-transferable utilities; agents from the
opposite sides form matching pairs (e.g., buyers-sellers) and negotiate the terms of their math which may
include a monetary transfer. Competitive equilibria are the elements of the core of this game.

We present the first combinatorial characterization of competitive equilibria that relates the utility of
each agent at equilibrium to the equilibrium utilities of other agents in a strictly smaller market excluding
that agent; thus automatically providing a constructive proof of existence of competitive equilibria in such
markets.

Our characterization also yields a group strategyproof mechanism for allocating indivisible goods to unit
demand buyers with non-quasilinear utilities that highly resembles the Vickrey--Clarke–-Groves (VCG)
mechanism. As a direct application of this, we present a group strategyproof welfare maximizing mechanism for
Ad-Auctions without requiring the usual assumption that search engine and advertisers have consistent
estimates of the clickthrough rates.

\end{abstract}

\pagenumbering{roman} \thispagestyle{empty}

%
%
%
%
%
%
%



\newpage
\pagenumbering{arabic}

\section{Introduction}

In markets with a common currency it is typical to assume transferable utilities, i.e., that an agent can
transfer some of its utility to another agent at no loss;  however this assumption rarely holds in practice
as wealthy agents often derive less marginal utility from a fixed increase to their wealth compared to poor
agents. The situation is more pronounced when monetary transfers are significant compared to agents' wealth.
We consider non-transferable utilities in the context of matching markets.

In this paper we consider a two-sided matching market with non-transferable utilities. Agents are divided in
two sides (e.g., sellers-buyers, men-women, etc). Each agent can partner with at most one other agent from
the opposite side, and they mutually decide about the terms of their partnership which may involve a monetary
transfer.  We are interested in the ``\emph{stable}'' outcomes of this market, i.e., outcomes in which no
coalition of agents could rearrange in such a way that would strictly benefit every member of the coalition.
Any such stable outcome corresponds to a \emph{competitive equilibrium}.

A competitive equilibrium is a feasible matching of agents (including the terms of matching and monetary
transfers) such that (i) there are no agents who could form a matching pair in such a way that would benefit
both of them better than their current state, and (ii) there is no matched agent who would prefer to be
unmatched.

Throughout this paper we consider markets in canonical form consisting of buyers on one side and goods (or
sellers) on the opposite side such that (i) utility of each buyer depends on the choice of good received and
the payment charged (utilities are decreasing in payment but not necessarily quasilinear), and (ii) utility
of each seller is equal to the payment received.  A competitive equilibrium then corresponds to an assignment
of prices to goods together with  envy--free matching of goods to buyers, i.e., such that any buyer receives
the good that she prefers the most at the assigned prices. In particular, any unassigned good should have a
price of zero. Any market can be represented in the canonical form by treating the utilities of agents of one
side of the market as prices of goods (i.e., each agent represented by a good), and by considering agents of
the opposite side as buyers;  the utility of a buyer for a good at a price $x$ then corresponds to the
highest utility the agent represented by the buyer could get if matched with the agent represented by the
good under any set of terms or monetary transfer that would give the agent represented by the good a utility
of at least $x$.

Competitive equilibria and their various properties have been studied extensively in the past, e.g.,
\citet{DG85}, \citet{G84}, \citet{Q82}. It has been shown that a competitive equilibrium always exists,
however previous proofs are non-constructive and provide little insight into equilibrium structure. There has
been independent recent work on algorithmic and/or graph theoretic characterizations of competitive
equilibria, e.g., \citet{CL10a,CL10b} and \citet{MOP10a, MOP10b}; which have been focused on obtaining
iterative processes for bottom-up construction of competitive equilibrium by growing an indifference graph
while making gradual adjustments to prices. The current paper provides a simple and natural characterization
of competitive equilibria based on a top-down construction, also yielding a simple recursive algorithm for
computing them; it also provides deep insight into equilibrium structure and how it reacts to changes in
supply/demand.

A competitive equilibrium can be concisely represented by its price vector. A price vector induces a utility
for each buyer which is equal to the highest utility the buyer could get from any of the goods at the
specified prices.  At a competitive equilibrium, utilities of buyers are exactly equal to their induced
utilities from the equilibrium prices. Similarly, a utility vector induces a price on each good, which is the
highest price any buyer would pay for that good given the specified utility vector. \citet{G84} showed that
the set of equilibrium price vectors  form a complete lattice (similarly the set of equilibrium utility
vectors). In particular, there is a competitive equilibrium with the highest prices (and lowest utilities),
and  there is one with the lowest prices (and highest utilities).

In this paper we present a simple combinatorial characterization of competitive equilibria in two-sided
matching markets with non-transferable utilities. Our characterization is inductive, i.e., it relates the
equilibrium prices/utilities to those of strictly smaller markets, thereby providing a constructive proof of
existence of competitive equilibria, and yielding a purely combinatorial approach for computing the
equilibria. We summarize our results below:
\begin{itemize}
\item

The utility of a buyer at the lowest priced competitive equilibrium is equal to her induced utility from the
prices of the highest priced competitive equilibrium of the market without her.  Similarly the price of a
good at the highest priced competitive equilibrium is equal to its induced price from the utilities of the
lowest prices competitive equilibrium of the market without that good. Combining the previous two
characterization implies that the price a buyer pays for a good she receives at the lowest priced competitive
equilibrium is equal to the highest price the rest of the market would be willing to pay for that good, i.e.,
it is the price at which the rest of the market becomes indifferent between buying and not buying the good.
The previous characterization also yields a mechanism for obtaining the lowest priced competitive equilibrium
which highly resembles the Vickrey--Clarke–-Groves (VCG) mechanism, albeit for non-transferable utilities.

\item
There is a continuum of competitive equilibria between the lowest priced competitive equilibrium and the
highest priced competitive equilibrium. In particular, given any lower bound on the price vector or utility
vector,  one could obtain a competitive equilibrium satisfying the bounds, if any, using a similar
combinatorial characterization as above.

\item
At the lowest priced competitive equilibrium, every set of goods with strictly positive prices is weakly over
demanded, i.e., there is a buyer who is not matched to any of the goods in that set, but is weakly interested
in at least one of the goods in that set. Similarly, at the highest priced competitive equilibrium, in every
set of buyers with strictly positive utilities there is a buyer who is weakly interested in a good that is
not among the goods matched to that set of buyers. The next two properties can be derived from the previous
two properties.

\item
At the highest priced competitive equilibrium, the market is indifferent between keeping or losing any single
good. In other words, if any single good is removed from the market, there is matching of the remaining goods
to buyers at the current prices such that the market remain at equilibrium with the same prices/utilities as
before. Similarly, at the lowest priced competitive equilibrium, the market is indifferent between keeping or
losing any single buyer. In other words, if any single buyer is removed from the market, there is matching of
the goods to remaining buyers at the current prices such that the market remain at equilibrium with the same
prices/utilities as before.

\end{itemize}

\paragraph{Roadmap}
We start by defining our model and assumptions in \autoref{sec:prelim}.  In \autoref{sec:challenge}, we
provide a brief overview of competitive equilibria and their structure in the case of transferable utilities,
we then identify the main difficulties of applying existing techniques to non-transferrable utilities. We
then present our main results in  \autoref{sec:main} . Finally, in \autoref{sec:aa} we present an immediate
application of our results to Ad-Auctions. The next section provides a quick review of the related
literature.

\subsection{Related Work}
The problem we consider is a one-to-one matching with  non-transferable utilities as described
by~\citet{DG85}. The problem is a generalization of the assignment game of~\citet{SS71} to non-quasilinear
utilities. \citeauthor{SS71} studied this problem for quasilinear utilities, proved that core outcomes always
exist, and showed that they form a lattice. Note that core outcomes correspond exactly to competitive
equilibria in these models. The existence of competitive equilibria for non-quasilinear utilities was proved
by~\citet{Q82} and by~\citet{G84}. \citeauthor{Q82} showed that the game defined by this model is a
\emph{``Balanced Game''}, where for general n-person balanced games \citet{S67} had shown that the core is
non-empty. Using a similar approach \citet{KY86} proved a similar result for a slight generalization of this
problem. \citet{G84} also showed that a competitive equilibrium always exists using combinatorial topology.
The proof of \citeauthor{G84} is based on a generalization of the KKM lemma (see \citet{KKM29}) which is the
continuous variant of the Sperner's lemma. Both of these proofs are non-constructive and only show the
existence of an equilibrium. As such, they don't provide a natural characterization of equilibrium utilities
\footnote{The proof of \citeauthor{S67} provides an algorithm based on the pivoting algorithm of
\citet{LH64}, which can be combined with the construction of \citeauthor{Q82} to yield an algorithm that
computes a competitive equilibrium in $2^{\BigO{n!}}$ iterations running on a matrix with $\BigO{n!}$
columns. Nevertheless, the resulting algorithm is more of an exhaustive search and does not provide any
insight into the equilibrium structure.}. \citet{L91} showed that in one-to-one markets with quasilinear
utilities, prices at the lowest competitive equilibrium equal VCG payments. The proof of \citeauthor{L91} is
based on writing the LP for computing the welfare maximizing allocation(i.e., a maximum weight matching) and
showing that every optimal assignment of the dual variables correspond exactly to the prices/utilities at a
competitive equilibrium. The proof crucially depends on the quasilinearity of utilities. \citet{DG85} studied
various properties of these markets in the case of non-quasilinear utilities. They showed that the set of
competitive equilibria is a lattice. \citet{DGS86} later proposed an ascending auction for computing the
competitive equilibrium with the lowest prices in the case of quasilinear utilities. Interestingly, their
method is the same as the Hungarian method (see \citet{K56}) for finding a maximum weight matching. Notice
that at this point there was still no combinatorial characterization of competitive equilibria for the case
of non-quasilinear utilities and no constructive proof of existence except for the ascending auction of
\citet{AG89} for piece-wise linear utility functions.

There has been independent recent work on algorithmic and/or graph theoretic characterizations of competitive
equilibria by \citet{CL10a,CL10b} and \citet{MOP10a, MOP10b}; which have been focused on obtaining iterative
processes for bottom-up construction of competitive equilibrium by growing an indifference graph while making
gradual adjustments to prices.

The are closely related works involving market with many-to-one or many-to-many matchings, either with
quasilinear utilities or with ordinal preferences. For competitive equilibria in markets with many-to-one
matchings and with quasilinear utilities see \citet{bm97,GS99,AM02}. For stable (core) allocations in markets
with many-to-one or many-to-many matchings with ordinal preferences see ~\citet{GS62,KC82,HM05,HK10}.

\section{Preliminaries}
\label{sec:prelim}%

\paragraph{Model}
Let $\Market=(\Buyers, \Goods, \UFun)$ denote a market where $\Buyers$ is a set of unit demand buyers,
$\Goods$ is a set of goods, and $\UFun_i^j(x)$ denote the utility of buyer $i \in \Buyers$ for receiving good
$j \in \Buyers$ and making a payment of $x$. The utility functions are private information of the respective
buyers. We assume that every $\UFun_i^j(x)$ is continuous and decreasing and there exists a large enough
$\MaxVal_i^j \in \PosReals$ such that $\UFun_i^j(\MaxVal_i^j) \le 0$. Without loss of generality, we assume
that $\UFun_i^j(x)$ is defined for all $x \in \Reals$ and its range also covers the whole $\Reals$
\footnote{Suppose $\MaxVal_i^j$ is the smallest non-negative payment for which $\UFun_i^j(\MaxVal_i^j) \le
0$. Since at the equilibrium all prices/utilities are non-negative, we can redefine $\UFun_i^j$ for values of
$x$ outside of $\RangeAB{0}{\MaxVal_i^j}$ as follows without affecting the equilibrium: for $x < 0$ redefine
$\UFun_i^j(x)\leftarrow \UFun_i^j(0)-x$, and for $x> \MaxVal_i^j$ redefine $\UFun_i^j(x)\leftarrow\UFun_i^j(
\MaxVal_i^j)-(x-\MaxVal_i^j)$.}, and therefore $\UFun_i^j(x)$ is invertible everywhere. We will use
$\PFun_i^j$ to denote the inverse of $\UFun_i^j$, i.e., in order to give buyer $i$ a utility of $x$ from good
$j$, the buyer should be charged a payment of $\PFun_i^j(x)$. Observe that $\PFun_i^j(x)$ is also a
continuous and decreasing function whose domain/range is the whole $\Reals$. For the rest of this paper, we
adopt the notation of using $\UFun$ and $\PFun$ to denote functions returning utilities and prices, and
$\UVec$ and $\PVec$ to denote utility vector and price vector. We use subscripts to index agents and
superscripts to index goods (i.e., $\UVec_i$, $\PVec^j$, etc). Negative subscripts/superscripts are used to
exclude the specified index (e.g., $\UVec_{-i}$ is the same as $\UVec$ but with the $i^{th}$ entry removed).

\begin{definition}[Competitive Equilibrium (CE)]
\label{def:CE}%
A \emph{``Competitive Equilibrium''} (henceforth abbreviated as CE) is an assignment of prices to goods
together with a feasible matching of goods to buyers such that each buyer receives her most preferred good at
the assigned prices, and such that every unmatched good has a price of $0$. Formally, for a market
$\Market=(\Buyers, \Goods, \UFun)$ we say that $\EQ=(\UVec, \PVec)$ is a CE of $\Market$ with price vector
$\PVec$ and utility vector $\UVec$, if and only if there exists a \emph{``supporting matching''} $\MVec$ such
that:
\begin{itemize}
\item
For every buyer $i$ and good $j$, if $i$ and $j$ are matched in $\MVec$, then $\UVec_i = \UFun_i^j(\PVec^j)$,
otherwise $\UVec_i \ge \UFun_i^j(\PVec^j)$.
\item
Every buyer $i$ that is unmatched in $\MVec$ should have a zero utility (i.e., $\UVec_i=0$). Similarly, Every
good $j$ that is unmatched in $\MVec$ should have a zero price (i.e., $\PVec^j=0$).
\item All prices/utilities must be non-negative, i.e., $\UVec_i \ge 0, \PVec^j \ge 0$.
\end{itemize}
For notational convenience, we often use $\UFun(\EQ)$, $\PFun(\EQ)$ to refer to $\UVec$, $\PVec$
respectively. $\MVec(i)$ will denote the good matched to buyer $i$ (if unmatched, $\MVec(i)=\emptyset$), and
$\MVec^{-1}(j)$ will denote the buyer matched to good $j$ (if unmatched, $\MVec^{-1}(j)=\emptyset$).
\end{definition}


Note that any CE, say $\EQ$, can be specified by either $\UFun(\EQ)$ or $\PFun(\EQ)$; given either the price
vector or the utility vector, the other one can be uniquely computed by taking the induced prices/induced
utilities as defined next.

\begin{definition}[Induced utilities/Induced Prices]
Given a market $\Market=(\Buyers, \Goods, \UFun)$ and a price vector $\PVec$, we define $\UFun_i(\PVec)$ to
denote the utility induced on buyer $i$ by offering the goods at prices $\PVec$. We can formally define
$\UFun_i$ as follows.
\begin{align}
    \UFun_i(\PVec) & = \max(0, \max_{j \in \Goods} \UFun_i^j(\PVec^j)) \label{eq:ind:u}
\end{align}
Similarly, given a utility vector $\UVec$, we define $\PFun^j(\UVec)$ to denote the price induced on good $j$
which is the highest price any buyer would pay for good $j$ assuming their utilities are fixed at $\UVec$,
i.e.,
\begin{align}
    \PFun^j(\UVec) & = \max(0, \max_{i \in \Buyers} \PFun_i^j(\UVec_i)) \label{eq:ind:p}.
\end{align}
\end{definition}

It is easy to see that if $\EQ=(\UVec, \PVec)$ is a CE, then $\PVec^j = \PFun^j(\UVec)$ and $\UVec_i =
\UFun_i(\PVec)$\footnote{The inverse is not true, i.e., given non-negative vectors $\PVec$ and $\UVec$ such
that for all $i$ and $j$ we have $\PVec^j = \PFun^j(\UVec)$ and $\UVec_i = \UFun_i(\PVec)$, there may exist
no feasible matching $\MVec$ such that $(\UVec, \PVec)$ is a CE.}. Throughout the rest of this paper, we
often specify a CE such as $\EQ$ by specifying either $\PFun(\EQ)$ or $\UFun(\EQ)$; note that once $\PVec$
and $\UVec$ are determined, it is straightforward to find a supporting matching $\MVec$\footnote{If $\PVec$
and $\UVec$ correspond to a CE, then a supporting matching $\MVec$ can be computed as follows: consider the
bipartite graph in which there is an edge between each buyer $i$ and good $j$ iff $\UVec_i =
\UFun_i^j(\PVec^j)$ and find a maximum matching covering all vertices with strictly positive price/utility.}.
Note that there could be multiple feasible supporting matchings for a given set of prices/utilities; in that
case we could pick any such matching arbitrarily.

\section{Transferable vs. Non-Transferable Utilities}
\label{sec:challenge}%

We start by describing the structure of CEs for transferable utilities and then compare it with the more
general case of non-transferable utilities. We show that the standard approaches for analyzing and/or
computing CEs with transferable utilities fail to generalize to non-transferable utilities; in particular we
show that ascending auctions --- even when carefully modified to work for non-transferable utilities --- may
take arbitrarily long time to converge.

Consider a market $\Market=(\Buyers, \Goods, \UFun)$ with transferable utilities. Without loss of generality
we may assume that payments(prices) are also measured in the same units as the utilities, therefore we may
assume every utility function is of the form $\UFun_i^j(x)=\Val_i^j-x$, where $\Val_i^j$ can be thought of as
the valuation of the buyer $i$ for good $j$. Then every CE of the market corresponds to a maximum weighted
matching and a minimum weight covering of the bipartite graph consisting of buyers/goods in which the edge
between buyer $i$ and good $j$ has weight $\Val_i^j$. The proof of this claim follows from the linear
programming relaxation of the maximum weight matching problem and its dual as shown below. In the following,
$\Alloc_i^j$ is the variable corresponding to the allocation of good $j$ to buyer $i$.

\begin{align*}
    &\begin{aligned}
    \text{\textbf{(Primal)}} && \\
    \text{maximize} &\,& &\sum_{i \in \Buyers} \sum_{j \in \Goods} \Val_i^j \Alloc_i^j \\
    \text{subject to}   && & \textstyle \sum_{j \in \Goods} \Alloc_i^j \le 1, & & \forall i \in \Buyers  \\
                        && & \textstyle \sum_{i \in \Buyers} \Alloc_i^j \le 1, & & \forall j \in \Goods \\
                        && & \Alloc_i^j \ge 0 \\
    \end{aligned}
    \BREAKCOL
    &\begin{aligned}
    \text{\textbf{(Dual)}} &&\\
    \text{minimize} &\,&& \sum_{i\in \Buyers} \UVec_i + \sum_{j\in \Goods} \PVec^j \\
    \text{subject to}   && & \UVec_i + \PVec^j \ge \Val_i^j,  & & \forall i \in \Buyers, \forall j \in \Goods \\
                        && & \UVec_i \ge 0 \\
                        && & \PVec^j \ge 0
    \end{aligned}
\end{align*}

The optimal assignments of the above linear program and its dual correspond to CEs as follows. Consider an
extreme point optimal assignment of the primal and any optimal assignment for the dual. The extreme points of
the matching polytope are integral so every $\Alloc_i^j$ is either $0$ or $1$. By strong duality both the
primal and the dual have the same optimal value which happens to be equal to the optimal social welfare. By
complementary slackness, if $\Alloc_i^j=1$, then the corresponding dual constraint must be tight which
implies $\UVec_i = \Val_i^j-\PVec^j$. Similarly, complementary slackness implies that every good with a
non-zero price (similarly ever buyer with a non-zero utility) should be matched in the primal. Finally,
observe that the dual constraints ensure that the utility of any buyer from her current assignment is the
maximum she can get from any good under the current prices.

The previous argument also implies that every CE yields a welfare maximizing allocation. Furthermore, among
all CEs (i.e., all optimal assignments of the dual program), the one with the minimum prices coincides with
the outcome of the Vickrey--Clarke–-Groves (VCG) mechanism.

A maximum weight matching can alternatively be computed using the \emph{Hungarian Method} of \citet{K56}.
Interestingly, the Hungarian method is equivalent to the following ascending price auction proposed by
\citet{DGS86} for computing the CE with the lowest prices:

\begin{definition}[Ascending Price Auction]
\label{auc:hung}%
Set all the prices equal to $0$. Find a minimally over demanded subset of goods at the current prices, i.e.,
a subset $\GSubsetT$ of goods such that there is a subset $\ISubsetS$ of the buyers who strictly prefer one
or more of the goods in $\GSubsetT$ to goods outside of $\GSubsetT$ at the current prices and such that
$\Abs{\ISubsetS} > \Abs{\GSubsetT}$; increase the prices of goods in $\GSubsetT$ simultaneously (at a uniform
rate) until one of the buyers in $\ISubsetS$ becomes indifferent between a good outside of $\GSubsetT$ and
her preferred good(s) in $\GSubsetT$; at that point, recompute the minimally over-demanded subset and repeat
the process until there is no over demanded subset of goods.
\end{definition}

Unfortunately none of the above approaches can be effectively generalized to handle non-transferable
utilities. The linear programming relaxation and the VCG mechanism heavily rely on quasi-linearity of
utilities. Furthermore, the ascending auction -- even if modified to work with non-transferable utilities --
may take arbitrarily long time to reach equilibrium. Constructive approaches have been proposed for special
classes of non-transferable utilities (e.g., \citet{AG89}), all of which essentially boil down to an
ascending auction of the following form: raise the prices at some rate (potentially non-uniform) to the next
point at which there is a change in the demand structure \footnote{Note that the prices can be jumped
(discretely) to the next point at which there is a change in the demand sets.}; then, recompute the demand
structure and repeat. For quasilinear utilities, the ascending auction stop after
$\BigO{\Abs{\Buyers}+\Abs{\Goods}}$ iterations (a new iteration starts each time the demand sets change). For
special cases of non-quasilinear utilities (for example quasilinear but with a hard budget constraints)
similar ascending auctions have been applied. However such auctions may not reach an equilibrium in finite
time as we illustrate with an example next.

The main problem with ascending auctions for non-transferable utilities occurs when the prices of an over
demanded subset of goods, say $\GSubsetT$, is to be raised. For transferable (quasilinear) utilities, the
prices of all of the goods in $\GSubsetT$ are raised at the same rate without affecting the relative
preferences of buyers over the goods in $\GSubsetT$; however, that is not the case for non-transferable
(non-quasilinear) utilities. In the general case, one may need to raise the prices of goods in $\GSubsetT$ at
different and possibly variable rates and even then the preferences of buyers over the goods in $\GSubsetT$
may change an unbounded number of times. We demonstrate the problem in the following example:

\begin{example}
\label{ex:1}%
Suppose there are 3 goods and 4 buyers whose utility functions are given in the following table in which
$\Val$ is some constant (at least $2$) and $x$ is the payment(price):

\begin{centering}
\small
\begin{tabular}{|l|l|l|l|}
\hline
       & good 1    & good 2    & good 3 \\
\hline
buyer 1 & $\UFun_1^1(x)=\Val+1-x$ & $\UFun_1^2(x)=\Val+1-x$   & $\UFun_1^3(x)=\Val+1-x$  \\
\hline
buyer 2 & $\UFun_2^1(x)=0-x$       & $\UFun_2^2(x)=\Val+1-x$ & $\UFun_2^3(x)=0-x$ \\
\hline
buyer 3 & $\UFun_3^1(x)=0-x$       & $\UFun_3^2(x)=0-x$       & $\UFun_3^3(x)=\Val+1-x$ \\
\hline%
buyer 4 & $\UFun_4^1(x)=\Val-x$   & $\UFun_4^2(x)=\Val-x-\frac{\Val-x}{\Val}\sin(\Val\log(\Val-x))$  &
$\UFun_4^3(x)=\Val-x-\frac{\Val-x}{\Val}\cos(\Val\log(\Val-x))$ \\ \hline
\end{tabular}
\end{centering}

All buyers have quasilinear utilities except buyer $4$. Moreover, all utilities are strictly decreasing in
payment (assuming $\Val\ge 2$). \autoref{fig:ascending_auction} shows the prices of goods throughout the
ascending auction. We should emphasis that in this particular example the ascending path of prices is unique.
The ascending auction can only increase the prices of goods that are over demanded, i.e., demanded by at
least two buyers. Furthermore, it can only raise the price of a good to the point where the demand of that
good is about to drop to $1$. Therefore, for every good with a positive price there should be at least a
(weak) demand of $2$ at any point in the auction. Observe that the demand set of buyer $1$ and $4$ changes an
infinite number of times during the ascending auction. Specifically, the demand set of both buyer $1$ and $4$
include good $1$ at all times. However, the demand set of buyer $1$ includes good $2$ and/or good $3$ only at
the times the price curves of those goods overlap with the price curve of good $1$. Similarly, the demand set
of buyer $4$ includes good $2$ and/or good $3$ only at the times the price curves of those goods do not
overlap with the price curve of good $1$. Observe that the demand structure changes an infinite number of
times as the price of the goods approach $\Val$. So an ascending auction does not stop in finite time.
\end{example}

\begin{figure}[h]
\centering
\includegraphics[height=0.4\textheight]{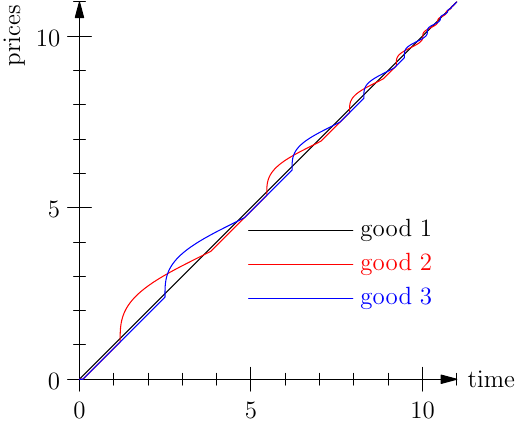}\\
{\tiny The blue and red curves have been slightly shifted down to make the black line visible.}
\caption{prices of goods in the ascending auction of \autoref{ex:1} assuming that $\Val=11$ and assuming the
price of good $1$ increases at the rate of $1$ .\label{fig:ascending_auction}}
\end{figure}

The previous example, although contrived, illustrates a fundamental problem that arises with ascending
auctions and constructive proofs that are based on them. In general, ascending auctions are quite sensitive
to the structure of utility functions. One of the main contributions of the current paper is a direct
approach for computing the lowest CE without running an ascending auction (see \autoref{thm:ind}).

\textbf{Quasilinear with hard budget limit.} In this case the utility function is quasilinear as long as the
payment does not exceed the budget limit; and if the payment exceeds the budget, the utility function goes to
negative infinity~\footnote{This ensures that an individually rational mechanism never charges the buyer more
than her budget limit.}. Notice that the issue outlined in the previous example does not arise with
quasilinear utilities with hard budget constraints. In fact, ascending auctions with hard budget constraints
converge almost as fast as ascending auction with quasilinear utilities and no budget constraints because
each buyer may hit her budget limit at most $\Abs{\Goods}$ times (once per each good) and beyond that point
she never demands the same good again. For the same reason, quasilinear utilities with hard budget limits are
easier to deal with than general non-quasilinear utilities.



\section{Characterization of Competitive Equilibria}
\label{sec:main}%

This section studies structural properties of CEs. The main result of this section is \autoref{thm:ind} which
characterizes the equilibrium prices/utilities.


We start by showing that the set of CEs is a complete lattice. Consider a market $\Market=(\Buyers, \Goods,
\UFun)$, and let $\EQ=(\UVec, \PVec)$ and $\EQ'=(\UVec', \PVec')$ be two arbitrary CEs of $\Market$ with
supporting matchings $\MVec$ and $\MVec'$ respectively. We define the $\min$ and $\max$ operators for CEs as
follows.

\begin{align*}
    \min(\EQ,\EQ') & =(\max(\UVec,\UVec'), \min(\PVec,\PVec')) \qquad&
    \text{with supporting matching} \quad \MVec'' &=
    \begin{cases}
        \MVec(i) & \UVec_i \ge \UVec'_i \\
        \MVec'(i) & \UVec_i < \UVec'_i
    \end{cases}\\
    \max(\EQ,\EQ') & =(\min(\UVec,\UVec'), \max(\PVec,\PVec')) \qquad&
    \text{with supporting matching} \quad \MVec''' &=
    \begin{cases}
        \MVec'(i) & \UVec_i \ge \UVec'_i \\
        \MVec(i) & \UVec_i < \UVec'_i
    \end{cases}
\end{align*}

Note that we assume $\min$ and $\max$ operators applied to vectors (e.g., price vector, utility vector)
return respectively the component-wise minimum and maximum of those vectors. The following theorem was
originally proved by~\citet{DG85}.

\begin{theorem}[Equilibrium Lattice~\citet{DG85}]
\label{thm:lattice}%
Given a market $\Market=(\Buyers, \Goods, \UFun)$, if $\EQ$ and $\EQ'$ are any two CEs of $\Market$, then
$\EQ''=\min(\EQ,\EQ')$ and $\EQ'''=\max(\EQ,\EQ')$ are also CEs of $\Market$. Consequently the set of all CEs
is a complete lattice and has a unique minimum and a unique maximum~\footnote{That is because the set of all
CEs is a closed and compact set, which can be proved by applying the Bolzano–Weierstrass theorem and using
the fact that the set of all possible matchings is finite.}.
\end{theorem}

We will refer to the CE with the lowest prices as the \emph{lowest CE} and the one with the highest prices as
the \emph{highest CE}. Observe that the lowest CE has the highest utilities and the highest CE has the lowest
utilities. Throughout the rest of this paper we implicitly use the lattice structure of the set of CEs
without making explicit references to \autoref{thm:lattice}.


\citet{Q82} and \citet{G84} presented existential arguments (non-constructive) proving that set of CEs is
always non-empty. The main contribution of this paper is an inductive characterization of equilibria that
relates the prices/utilities of the highest/lowest CE of a market to those of a strictly smaller market by
removing either a buyer or a good. This inductive characterization automatically yields a constructive proof
of the non-emptiness of the set of CEs.

For any given market $\Market=(\Buyers, \Goods, \UFun)$, let $\Market_{-i}=(\Buyers-\{i\},\Goods,\UFun)$
(i.e., the same market without buyer $i$), and let $\Market^{-j}=(\Buyers,\Goods-\{j\},\UFun)$ (i.e. the same
market without good $j$). The following theorem relates the equilibrium price/utility of any good/buyer to
those of a strictly smaller market which excludes the respective good/buyer; hence yielding an inductive
approach for computing equilibrium prices/utilities.

\begin{theorem}[Inductive Equilibrium Characterization]
\label{thm:ind}%
Consider an arbitrary market $\Market=(\Buyers, \Goods, \UFun)$. Let $\EQL$ be the lowest CE of $\Market$,
$\EQH$ be the highest CE of $\Market$, $\EQL^{-j}$ be the lowest CE of $\Market^{-j}$, and $\EQH_{-i}$ be the
highest CE of $\Market_{-i}$. Then

\renewcommand{\labelenumi}{\Roman{enumi}.}
\renewcommand{\theenumi}{\arabic{theorem}.\Roman{enumi}}

\begin{enumerate}
\item $\UFun_i(\EQL)=\UFun_i(p(\EQH_{-i}))$, \label{eq:i}
\item $\PFun^j(\EQH)=\PFun^j(u(\EQL^{-j}))$. \label{eq:ii}
\end{enumerate}
Furthermore,
\begin{enumerate}[resume]
\item $\PFun^j(\EQL) \le \PFun^j(\EQH_{-i})$, in particular, if $i$ and $j$ are matched in $\EQL$, then $\PFun^j(\EQL) = \PFun^j(\EQH_{-i})$, \label{eq:iii}
\item $\UFun_i(\EQH) \le \UFun_i(\EQL^{-j})$, in particular, if $i$ and $j$ are matched in $\EQH$, then $\UFun_i(\EQH) = \UFun_i(\EQL^{-j})$. \label{eq:iv}
\end{enumerate}
\renewcommand{\theenumi}{\Roman{enumi}}
\end{theorem}

Note that equilibrium prices/utility can be fully computed by recursive application of the first two
equations of the above theorem(recall that a CE can specified by either of its price vector or utility
vector). The result of the theorem can be interpreted as follows.

\begin{itemize}
\item
Equation \ref{eq:i}: utility of buyer $i$ at the lowest CE of market $\Market$ can be computed as follows.
Remove $i$ from the market. Compute the prices at the highest CE of the rest of the market. Offer those
prices to buyer $i$ and compute her utility from her most preferred good at those prices.

\item
Equation \ref{eq:ii}: The price of any good $j$ at the highest CE of the market $\Market$ can be computed as
follows. Remove good $j$ from the market. Compute the buyers' utilities at the lowest CE of the rest of the
market. Ask each buyer to name a price for good $j$ that would give her the same utility as what she
currently gets. Take the maximum among the named prices.
\end{itemize}

Combining the two equations of \autoref{thm:ind} yields the following characterization of payments.

\begin{corollary}
At the lowest CE of a market, the payment of each buyer for the good she has received is equal to the highest
price the rest of the market would be willing to pay for that good given their current utilities.
\end{corollary}

Notice the striking similarity between the above characterization and the payments in the
Vickrey--Clarke–-Groves (VCG) mechanism, i.e., that the payment of each buyer is equal to her externality on
the rest of the buyers; however recall that VCG cannot be applied here as it crucially depends on
quasi-linearity of utilities.


Next we show there is a continuum of CEs between the highest CE and the lowest CE. We also present a
inductive characterization for obtaining such CEs.

\begin{definition}[$(\UVecLB, \PVecLB)$-Bounded CE]
Given a market $\Market=(\Buyers, \Goods,\UFun)$ and a lower bound price vector $\PVecLB \in
\PosReals^\Goods$ and a lower bound utility vector $\UVecLB \in \PosReals^\Buyers$, we say that $\EQ$ is a
\emph{$(\UVecLB,\PVecLB)$-bounded CE} of $\Market$ iff $\EQ$ is a CE of $\Market$ and $\PFun(\EQ) \ge
\PVecLB$ and $\UFun(\EQ) \ge \UVecLB$.

%

Note that for a given $\UVecLB$ and $\PVecLB$, the $(\UVecLB,\PVecLB)$-bounded CEs of $\Market$ form a
complete sublattice of all the CEs of $\Market$. In particular, there is a lowest and a highest
$(\UVecLB,\PVecLB)$-bounded CE of $\Market$. Notice that for arbitrary $\UVecLB$ and $\PVecLB$, a
$(\UVecLB,\PVecLB)$-bounded CE may not necessarily exist.
\end{definition}

\begin{theorem}[Continuity]
\label{thm:cont}%
Assume a market $\Market=(\Buyers, \Goods, \UFun)$ with equal number of buyers and goods (i.e.,
$\Abs{\Buyers}=\Abs{\Goods}$) and lower bounds $\PVecLB \ge \ZeroVec$ and $\UVecLB \ge \ZeroVec$ on the
prices/utilities. If a $(\UVecLB, \PVecLB)$-bounded CE exists, then at the lowest $(\UVecLB,
\PVecLB)$-bounded CE there exists at least one good $j^* \in \Goods$ whose price is exactly equal to its
lower bound (i.e. equal to $\PVecLB^{j^*}$), and at the highest $(\UVecLB, \PVecLB)$-bounded CE there exists
at least one buyer $i^* \in \Buyers$ whose utility is exactly equal to her lower bound (i.e. equal to
$\UVecLB_{i^*}$).
\end{theorem}

The proof of the above theorem is based on defining a new market $\Market'=(\Buyers,\Goods, \UFun')$ with
transformed utility functions  ${\UFun'}_i^j(x) = \UFun_i^j(x+\PVecLB^j)-\UVecLB_i$. There is a one to one
correspondence between CEs of $\Market$ and $\Market'$~\footnote{This statement is not true if $\Market$ does
not have a  $(\UVecLB, \PVecLB)$-bounded CE.}. Given that there are equal number of buyers and goods,
applying \autoref{thm:ind} and removing a buyer (or a good) from $\Market'$ implies that there will be an
unmatched good (or unmatched buyer) in the remaining market which will have a zero price (or zero utility)
which then implies the claim in the original market.

The following corollary immediately follows the previous theorem by setting the lower bounds to zero.

\begin{corollary}
Given a market $\Market=(\Buyers, \Goods, \UFun)$ with $\Abs{\Buyers}=\Abs{\Goods}$:
\begin{itemize}
\item At the lowest CE, there is at least one good that has a price of $0$.
\item At the highest CE there is at least one buyer that has a utility of $0$.
\end{itemize}
\end{corollary}

As another immediate corollary of \autoref{thm:cont}, a continuum of CEs can be obtained as follows. Define
$\PVecLB(t) = (1-t)\PFun(\EQL)+t \PFun(\EQH)$. Now applying \autoref{thm:cont} and taking the lowest
$(\ZeroVec, \PVecLB(t))$-Bounded CE for $t \in \RangeAB{0}{1}$ yields a continuum of equilibria between the
lowest CE and the highest CE of $\Market$.

\begin{corollary}
Given a market $\Market=(\Buyers, \Goods, \UFun)$ with $\Abs{\Buyers}=\Abs{\Goods}$, there is a continuum of
equilibria between $\EQL$ and $\EQH$.
\end{corollary}


Throughout the rest of this section we present several theorems which capture important properties of
lowest/highest CEs. The next two theorems capture the demand structure in the lowest/highest CEs; they also
play a key role in the proof of \autoref{thm:ind}.

\begin{definition}[Demand Sets]
Consider a market $\Market=(\Buyers, \Goods, \UFun)$ and let $\EQ=(\UVec, \PVec)$ be a CE of $\Market$. For
each subset $\ISubsetS \subseteq \Buyers$ of buyers, we define the demand set of $\ISubsetS$ in $\EQ$ as
\begin{align*}
    \DemandSet_{\ISubsetS}(\EQ) =
        \left\{j \middle| \exists i \in \ISubsetS : j \in \argmax_{j' \in \Goods} \UFun_i^{j'}(\PVec^{j'}) \right\}.
\end{align*}
Similarly, for each subset $\GSubsetT \subseteq \Goods$ of goods, we define the demand set of $\GSubsetT$ in
$\EQ$ as
\begin{align*}
    \DemandSet^{\GSubsetT}(\EQ) =
        \left\{i \middle| \exists j \in \GSubsetT : i \in \argmax_{i' \in \Buyers} \PFun_{i'}^{j}(\UVec_{i'}) \right\}.
\end{align*}
\end{definition}

\begin{theorem}[Tightness]
\label{thm:tight}%
Consider a market $\Market=(\Buyers, \Goods, \UFun)$, and let $\EQ$ be a CE of $\Market$.
\begin{itemize}
\item
$\EQ$ is the lowest CE of $\Market$ iff $\Abs{\DemandSet^\GSubsetT(\EQ)} \ge \Abs{\GSubsetT}+1$ for every
subset $\GSubsetT \subseteq \Goods$ of goods with strictly positive prices (i.e., at least
$\Abs{\GSubsetT}+1$ buyers are (weakly) interested in $\GSubsetT$).
\item
$\EQ$ is the highest CE of $\Market$ iff $\Abs{\DemandSet_\ISubsetS(\EQ)} \ge \Abs{\ISubsetS}+1$ for every
subset $\ISubsetS \subseteq \Buyers$ of buyers with strictly positive utilities (i.e., buyers in $\ISubsetS$
are (weakly) interested in at least $\Abs{\ISubsetS}+1$ goods).
\end{itemize}
\end{theorem}

To get a better intuition of \autoref{thm:tight} suppose $\EQ$ is the lowest CE of a market $\Market$ and
suppose there is a subset $\GSubsetT$ of goods with strictly positive prices for which the statement of the
theorem fails to hold (i.e. $\Abs{\DemandSet^\GSubsetT(\EQ)} < \Abs{\GSubsetT}+1$); then, we could
conceptually decrease the prices of goods in $\GSubsetT$ down to the point where either a buyer out of
$\DemandSet^\GSubsetT(\EQ)$ becomes \underline{indifferent} between her current allocation and some good in
$\GSubsetT$; or one of the goods in $\GSubsetT$ hit the price of $0$. But then, we get a CE lower than $\EQ$
which contradicts $\EQ$ being the lowest CE. Despite the easy intuition the formal proof of this statement is
more involved.


The following is a direct consequence of \autoref{thm:tight}.

\begin{theorem}
\label{thm:mod}%
Consider a market $\Market=(\Buyers, \Goods, \UFun)$, and let $\EQ=(\UVec, \PVec)$ be a CE of $\Market$.
\begin{itemize}
\item
If and only if $\EQ$ is the highest CE of $\Market$, then $\EQ^{-j}=(\UVec, \PVec^{-j})$ is also a CE for the
market $\Market^{-j}$ for every $j \in \Goods$.

\item
If and only if $\EQ$ is the lowest CE of $\Market$, then $\EQ_{-i}=(\UVec_{-i}, \PVec)$ is also a CE for the
market $\Market_{-i}$ for every $i \in \Buyers$.
\end{itemize}
\end{theorem}

Intuitively, the above theorem says that at the highest CE, after removing any single good from the market,
the assignments can be modified to get a CE for the rest of the market with the same prices/same utility as
before; similarly, at the lowest CE, after removing any single buyer $i$ from the market, the assignments can
be modified to get a CE for the rest of the market with the same prices/same utilities.

\begin{proof}
We only prove the first statement since the proof of the second one is similar (completely symmetric).
Consider a supporting matching $\MVec$ for $\EQ$. If either $j$ is unmatched or the utility of buyer who is
matched to $j$ is $0$ we are done. Otherwise, we run the following process during which we maintain a subset
$\ISubsetS$ of buyers with strictly positive utilities such that from each buyer $i \in \ISubsetS$ there is
an alternating sequence of buyers/goods of the form $j'_1,i'_1,j'_2,i'_2,\ldots,j'_k,i'_k$ (of some length
$k$) where $j'_1=j$ and $i'_k=i$ and such that $i'_r$ is indifferent between $j'_r$ and $j'_{r+1}$ and
$\MVec(i'_r)=j'_r$, for every $r \in \RangeN{k}$. We initialize $\ISubsetS$ to be the singleton containing
the buyer that is matched to $j$. Since all buyers in $\ISubsetS$ have strictly positive utilities, they must
be weakly interested in at least $\Abs{\ISubsetS}+1$ goods. So there exists a buyer $i^* \in \ISubsetS$ who
is weakly interested in a good $j^*$ that is not currently matched to any of the buyers in $\ISubsetS$; if
$j^*$ is itself matched to another buyer with positive utility, we add that buyer to $\ISubsetS$ and repeat;
otherwise we assign $j^*$ to $i^*$ and switch the assignments along the alternating path from $i^*$ to $j$
which yields a supporting matching for $\EQ^{-j}$. Note that the above process always finds such a an
alternating path in at most $\Abs{\Buyers}-1$ iterations. The ``only if'' direction is trivial by applying
\autoref{thm:tight}.
\end{proof}


Next, we show that the lowest CE is group strategyproof for buyers, i.e., that no group of buyers can collude
in such a way that they all get strictly higher utilities (assuming no side payments).

\begin{theorem}[Group Strategyproofness]
\label{thm:gsp}%
Given a market $\Market=(\Buyers, \Goods, \UFun)$, a mechanism that solicits buyer's utility functions and
computes the lowest CE of $\Market$ is group strategyproof. Formally, for every subset $\ISubsetS \subseteq
\Buyers$ of buyers who collude and misreport their utility functions, if $\EQL'$ denotes the lowest CE with
respect to the reported utility functions, then there is at least one buyer $i\in \ISubsetS$ whose utility at
$\EQL'$ is no better than her utility at $\EQL$.
\end{theorem}
\begin{proof}
The proof is by contradiction. Let $\ISubsetS$ be the largest subset of buyers who can collude and possibly
misreport their utility functions such that all of them obtain strictly higher utilities. Let $\EQL$ be the
lowest CE of $\Market$ with respect to the true utility functions and let $\EQL'$ be the lowest CE with
respect to the reported utility functions assuming that buyers is $\ISubsetS$ have colluded. Let $\GSubsetT$
be the subset of the goods that are matched to $\ISubsetS$ at $\EQL'$. Since all the buyers in $\ISubsetS$
are achieving strictly higher utilities at $\EQL'$, they must all be matched at $\EQL'$ (i.e.
$\Abs{\GSubsetT}=\Abs{\ISubsetS}$) and the prices of the goods in $\GSubsetT$ should be strictly lower at
$\EQL'$. That means the goods in $\GSubsetT$ must have had strictly positive prices in $\EQL$. By applying
\autoref{thm:tight}, we argue that there must have been a subset $\ISubsetS'$ of buyers of size at least
$\Abs{\GSubsetT}+1$ who were weakly interested in some of the goods in $\GSubsetT$ at $\EQL$. Consequently
all of the buyers in $\ISubsetS'$ must be getting a strictly higher utilities at $\EQL'$ because the prices
of all the goods in $\GSubsetT$ are strictly lower. But $\ISubsetS'$ is larger than $\ISubsetS$ which
contradicts our assumption that $\ISubsetS$ was the largest set of buyers who could all benefit from
collusion.
\end{proof}

Finally, we preset the proof of our main theorem.


\makeatletter
\renewcommand{\p@enumii}{\theenumi.}
\renewcommand{\p@enumiii}{\p@enumii.\theenumii.}
\makeatother

\begin{proof}[Proof of \autoref{thm:ind}] \
We only prove \eqref{eq:i} and \eqref{eq:iii}. The proofs of \eqref{eq:ii} and \eqref{eq:iv} are completely
symmetric to the other two.

The plan of the proof is as follows. We remove an arbitrary buyer $i$ from the market and compute the highest
CE of the rest of the market. We then show that the prices at the highest CE of the market without $i$ leads
to a valid CE for the whole market (including buyer $i$) but with a possibly different matching. We also show
that the induced utility of buyer $i$ from these prices is the same as her utility at the lowest CE of the
whole market\footnote{Note that the resulting CE is not necessarily the lowest CE of $\Market$. Only the
utility of buyer $i$ is equal to his utility at the lowest CE of $\Market$.}. The detail of the construction
is as follows.

Choose an arbitrary buyer $i \in \Buyers$. Let $\Market_{-i}$ denote the market without buyer $i$ and let
$\EQH_{-i}$ be the highest CE of the market $\Market_{-i}$. Note that the market $\Market_{-i}$ is of size
$\Abs{\Buyers}+\Abs{\Goods}-1$ so by inductively applying \autoref{thm:ind} to $\Market_{-i}$ --- which is of
a strictly smaller size --- we can argue that there exists a CE for $\Market_{-i}$, thus the highest CE of
$\Market_{-i}$ is also well-defined. Let $\PVec=\PFun(\EQH_{-i})$ be the prices at $\EQH_{-i}$. We claim that
using the prices $\PVec$ for the market $\Market$ leads to a valid CE, which we denote by $\EQ$. In
particular, all the prices/utilities at $\EQ$ are the same as the prices/utilities at $\EQH_{-i}$ and the
utility of buyer $i$ is $\UFun_i(\PVec)$, however the matching of goods/buyers might be different. To obtain
a supporting matching for $\EQ$, we start with a supporting matching for $\EQH_{-i}$ and modify it as
follows. If $\UFun_i(\PVec) = 0$ then we can leave buyer $i$ unmatched and the matching does not need to be
changed. Otherwise, let $j$ be the good from which buyer $i$ achieves her highest utility at the current
prices, i.e. $j \in \argmax_{j' \in \Goods}\UFun_i^{j'}(\PVec^{j'})$. We match $j$ to buyer $i$ and apply
\autoref{thm:mod} to the market $\Market_{-i}$ to conclude that after removing good $j$ the matching for the
rest of the market can be modified appropriately to support a CE with the same prices/utilities as before.
Next, we prove each part of the theorem statement.
\begin{itemize}
\item
Proof of \eqref{eq:iii} (case 1) \fbox{\textbf{ $\PFun^j(\EQL) \le \PFun^j(\EQH_{-i})$}}: Notice that $\EQ$
is a CE of market $\Market$ which has the same prices as the $\EQH_{-i}$. The prices at the lowest CE of
$\Market$ are no more than the prices at $\EQ$ so $\PFun(\EQL) \le \PFun(\EQ) = \PFun(\EQH_{-i})$.

\item
Proof of \eqref{eq:iii} (case 2) \fbox{\textbf{$\PFun^j(\EQL) = \PFun^j(\EQH_{-i})$ when $\MVec(i)=j$}}:
Notice that since $i$ and $j$ are matched, if we remove both of them the rest of $\EQL$ is still a valid CE
for $\Market_{-i}^{-j}$. Let $\EQL_{-i}^{-j}$ denote the lowest CE of $\Market_{-i}^{-j}$. Note that both
$\EQL$ and $\EQL_{-i}^{-j}$ are valid CEs for $\Market_{-i}^{-j}$ but $\EQL_{-i}^{-j}$ is the lowest, so the
prices of goods $\Goods-\{j\}$ might only be lower at $\EQL_{-i}^{-j}$ and so the utilities of buyers
$\Buyers - \{i\}$ might only be higher at $\EQL_{-i}^{-j}$ and so the price induced by buyers $\Buyers-\{i\}$
on good $j$ might only be lower at $\EQL_{-i}^{-j}$ than the price induced by them on good $j$ at $\EQL$.
However, by applying \autoref{thm:ind} inductively on market $\Market_{-i}$ --- which is strictly smaller ---
and using \eqref{eq:ii}, we get that $\PFun^j(\EQH_{-i})$ is exactly the induced price of buyers
$\Buyers-\{i\}$ on good $j$ at $\EQL_{-i}^{-j}$. Therefore, $\PFun^j(\EQH_{-i})$ must be less than or equal
to the induced price on good $j$ at $\EQL$ which is itself less that or equal to $\PFun^j(\EQL)$. On the
other hand, from the previous paragraph we have $\PFun^j(\EQL) \le \PFun^j(\EQH_{-i})$, so the two must be
equal.

\item
Proof of \eqref{eq:i} \fbox{\textbf{$\UFun_i(\EQL)=\UFun_i(\PFun(\EQH_{-i}))$}}: If $i$ is matched with $j$
in $\EQL$ then $\UFun_i = \UFun_i^j(\PFun^j(\EQL))$ and by the previous statement $\PFun^j(\EQL) =
\PFun^j(\EQH_{-i})$. Therefore $\UFun_i(\EQL) = \UFun_i^j(\PFun^j(\EQH_{-i})) = \UFun_i(\PFun(\EQH_{-i}))$.
The last equality follows from the fact that we chose $j$ to be the good from which buyer $i$ obtains her
highest utility at prices $\PFun(\EQH_{-i})$.

\end{itemize}
\end{proof}

\section{Application to Ad-Auctions}
\label{sec:aa}%
In this section, we present a truthful mechanism for Ad-auctions that combines pay per click (a.k.a charge
per click (CPC)) advertisers and pay per impression (a.k.a charge per impression (CPM)) advertisers,
potentially with non-quasilinear utility functions. In particular, our mechanism is welfare maximizing and
group strategyproof regardless of whether the search engine and the advertisers have consistent beliefs about
the clickthrough rates.

\paragraph{Model}
Given a set of advertisers $\Buyers$ and a set of slots $\Goods$, the utility of advertiser $i$ for a click
on her ad being displayed on slot $j$ at a price of $x$ is given by $\UFun_i^j(x)$ which is a continuous and
decreasing function in $x$ \footnote{We assume that $\UFun_i^j(x)$ reaches $0$ for a high enough $x$.}. Note
that $x$ specifies payment per click for a CPC advertiser, or payment per impression for a CPM advertiser. We
say that a CPC advertiser $i$ has standard utility function if for all slots $j$: $\UFun_i^j(x) =
\CTR_i^j(\Val_i^j-x)$ in which $\Val_i^j$ is the advertiser's value for a click on slot $j$ and $\CTR_i^j$ is
the advertiser's estimate about her clickthrough rate (CTR), i.e., the probability that her ad is being
clicked on if displayed on slot $j$. We also say that a CPM advertiser $i$ has standard utility function if
for all slots $j$: $\UFun_i^j(x) = \Val_i^j-x$ in which $\Val_i^j$ is the advertiser's value for a click on
slot $j$. $\CTRS_i^j$ will denote the search engine's estimate of the CTR of advertiser $i$ on slot $j$; this
could potentially be different than $\CTR_i^j$ (i.e., advertisers and search engine may have different
beliefs). Furthermore, $\Val_i^j$ and $\CTR_i^j$ are advertiser's private information but $\CTRS_i^j$ is
publicly available.

Consider what happens if we applied VCG payments in this environment assuming there are only CPC advertisers
with standard utility function. The first of the following linear programs compute the welfare maximizing
allocation (based on advertisers' reports), and its dual computes prices/utilities.

\begin{align*}
    &\begin{aligned}
    \text{\textbf{(Primal)}} \\
    \operatorname{maximize} &\,& \sum_{i\in \Buyers} \sum_{j\in \Goods} & \CTR_i^j \Val_i^j \Alloc_i^j \\
    \text{subject to} && \textstyle \sum_{j \in \Goods} \Alloc_i^j & \le 1, && \forall i \in \Buyers,\\
                            && \textstyle \sum_{i \in \Buyers} \Alloc_i^j & \le 1, && \forall j \in \Goods, \\
                            && \Alloc_i^j & \ge 0
    \end{aligned}
    \BREAKCOL
    &\begin{aligned}
    \text{\textbf{(Dual)}} \\
    \operatorname{minimize} &\,& \sum_{i\in \Buyers} \UVec_i &+ \sum_{j\in \Goods} \PVec^j \\
    \text{subject to} && \UVec_i + \PVec^j & \ge \CTR_i^j \Val_i^j, && \forall i \in \Buyers, \forall j \in \Goods \\
                            && \UVec_i & \ge 0 \\
                            && \PVec^j & \ge 0
    \end{aligned}
\end{align*}

The set of solutions to the dual program would be the set of CEs of the market and the one with the lowest
prices would correspond to the lowest CE which would also coincide with the VCG payments/utilities. However,
the problem is that payments must be charged \emph{per click} while $\PVec^j$ represents the expected payment
\emph{per impression}. So, per click payments would be given by dividing $\PVec^j$ by the probability of a
click which is the CTR. Observe that dividing by $\CTR_i^j$ breaks the strategyproofness guarantee because
$\CTR_i^j$ is reported by the advertiser and reporting a higher $\CTR_i^j$ would give the advertiser a higher
chance of winning a better slot while at the same time it would lower the payment. Dividing by $\CTRS_i^j$
also breaks the strategyproofness when $\CTRS_i^j$ and $\CTR_i^j$ are different. Therefore, the
straightforward application of VCG payments fails.

Consider the Ad-Auction problem as a two sided matching market with ads on one side and slots on the opposite
side. For a CPC advertiser $i$, allocating slot $j$ and charging a payment of $x$ per click, yields an
expected utility of $\UFun_i^j(x)$ for the advertiser and generates an expected revenue of $x/\CTRS_i^j$ per
impression for the search engine (from the perspective of the search engine based on its own estimated CTRs).
For a CPM advertiser $i$, allocating slot $j$ and charging a payment of $x$ per impression, yields an
expected utility of $\UFun_i^j(x)$ for the advertiser and generates a revenue of $x$ per impression for the
search engine. Consider a CE in this market (treat each slot as an independent agent even though all slots
are owned by the search engine). Our proposed mechanism solicits advertisers' utility functions and outputs
as its outcome the CE that has the highest advertiser utilities.

\begin{definition}(Ad-Auction Competitive Equilibrium)
\label{def:aa}%
The mechanism solicits advertisers' utility functions $\UFun_i^j$ and defines
$\Market'=(\Buyers,\Goods,\UFun')$, where for a CPC advertiser the utility function is given by
${\UFun'}_i^j(x)=\UFun_i^j(x/\CTRS_i^j)$, and for a CPM advertiser the utility function is given by
${\UFun'}_i^j(x)=\UFun_i^j(x)$. Compute the lowest CE of $\Market'$, call it $\EQL=(\UVec,\PVec)$, and let
$\MVec$ be a supporting matching. Allocate to advertiser $i$ the slot $\MVec(i)$. If $i$ is a CPC advertiser,
charge her $\PVec^{\MVec(i)}/\CTRS_i^j$ if there is a click on the slot; otherwise $i$ is a CPM advertiser
and should be charged $\PVec^{\MVec(i)}$ for an impression.
\end{definition}

Recall that $\Market'$ is the reduced form of the original market which is obtained by redefining the utility
functions of one side in terms of the utilities of the other side. Observe that $\EQL$ satisfies the
following inequality for each CPC advertiser $i$:
\begin{align*}
    \UVec_i &= \max_{j\in \Goods} \UFun_i^j(\PVec^j/\CTRS_i^j) & \forall j\in \Goods.
\end{align*}
it also satisfies the following inequality for each CPM advertiser $i$:
\begin{align*}
    \UVec_i &= \max_{j\in \Goods} \UFun_i^j(\PVec^j) & \forall j\in \Goods.
\end{align*}
$\PVec^j$ can be interpreted as the \emph{virtual price} of slot $j$, where as $\PVec^j/\CTRS_i^j$ can be
interpreted as the weighted price of slot $j$ for a CPC advertiser $i$. Observe that the virtual price
$\PVec^j$ represents the expected per impression payment required for slot $j$ computed from the perspective
of the search engine. Notice that the above inequalities ensure that each advertiser receives their preferred
slot according to their weighted prices.

The above mechanism can also be conceptually reinterpreted as an ascending auction in which all virtual
prices $\PVec_i^j$ start from $0$ and are gradually increased as long as there is excess demand; and such
that at any time during the auction a CPC advertiser $i$ observes a price of $\PVec^j/\CTRS_i^j$ for slot $j$
whereas a CPM advertiser $i$ observes a price of $\PVec^j$ for slot $j$.

Next theorem summarizes the important properties of the above mechanism.

\begin{theorem}
\label{thm:aa}%
Mechanism of \autoref{def:aa} is group strategyproof and also maximizes the social welfare in the following
sense. Let $\ISubsetS$ be a group of advertisers with standard utility functions who also agree with the
search engine on the CTRs and let $\ISubsetS'$ be the rest of the advertisers. Let $s$ denote the search
engine. Welfare of $\{s\} \cup \ISubsetS$ is always maximized. In particular, if all advertisers have
standard utility functions and agree with the search engine on the CTRs, then the outcome of this mechanism
coincides exactly with the VCG outcome.
\end{theorem}

Notice that mechanism \ref{def:aa} is group strategyproof regardless of whether the search engine and
advertisers have the same estimates about the clickthrough rates.

\section{Acknowledgement}
We thank Rakesh Vohra, Lawrence Ausubel, John Hatfield, Scott Kominers, Jason Hartline, Nicole Immorlica,
Sebastien Lahaie, David Easley, Larry Blume and Martin Pal for helpful discussions and/or comments.

\bibliography{ce}
\appendix

\section{Other Results and Omitted Proofs}
\label{sec:proofs}%

This section presents the proofs of the theorems which were omitted from the previous sections along with a
few lemmas which are used in those proofs.

\begin{lemma}[Entanglement]
\label{lem:entangle}%
Consider a market $\Market=(\Buyers, \Goods, \UFun)$. If there exists a CE of $\Market$ at which buyer $i$ is
matched with good $j$, then at any other CE of $\Market$ the price of good $j$ is higher if and only if the
utility of buyer $i$ is lower and vice versa. Note that this statement is true regardless of wether buyer $i$
and good $j$ are actually matched to each other in other CEs.

\end{lemma}


\begin{lemma}[Conservation of Matching]
\label{lem:match}%
Given a market $\Market=(\Buyers, \Goods, \UFun)$, for any $i \in \Buyers$, if there exists a CE of $\Market$
at which buyer $i$ has a strictly positive utility, then buyer $i$ is never unmatched in any CE of $\Market$.
Similarly, for any $j \in \Goods$, if there exists a CE of $\Market$ at which good $j$ has a strictly
positive price, then good $j$ is never unmatched at any CE of $\Market$.
\end{lemma}

\begin{proof}[Proof of \autoref{lem:entangle} and \autoref{lem:match}]
Let $\EQ$ be a CE of $\Market$ at which the hypothesis of the lemma holds, and let $\EQ'$ be any other CE of
$\Market$. Partition the buyers to three groups $\ISubsetS$, $\ISubsetS'$, $\ISubsetS''$, such that buyers in
$\ISubsetS$ have higher utilities at $\EQ$, buyers in $\ISubsetS'$ have higher utilities at $\EQ'$, and
buyers in $\ISubsetS''$ have the same utilities at $\EQ$ and $\EQ'$. Similarly, partitions the goods to
$\GSubsetT$, $\GSubsetT'$, $\GSubsetT''$, such that goods in $\GSubsetT$ have higher prices at $\EQ$, goods
in $\GSubsetT'$ have higher prices and $\EQ'$, and goods in $\GSubsetT''$ have the same prices at both $\EQ$
and $\EQ'$. The following statements are easy to derive using the definition of CE and using the fact that
both $\EQ$ and $\EQ'$ are CEs:

\begin{itemize}
\item At $\EQ$, all buyers in $\ISubsetS$ must be matched to goods in $\GSubsetT'$ so $\Abs{\ISubsetS} \le \Abs{\GSubsetT'}$.
\item At $\EQ'$, all goods in $\GSubsetT'$ must be matched to buyers in $\ISubsetS$ so $\Abs{\GSubsetT'} \le \Abs{\ISubsetS}$.
\end{itemize}

From the above statement, we can conclude $\Abs{\ISubsetS}=\Abs{\GSubsetT'}$ and buyers in $\ISubsetS$ and
goods in $\GSubsetT'$ must be matched to each other at both equilibria. Similarly:

\begin{itemize}
\item At $\EQ$, all goods in $\GSubsetT$ must be matched to buyers in $\ISubsetS'$ so $\Abs{\GSubsetT} \le \Abs{\ISubsetS'}$.
\item At $\EQ'$, all buyers in $\ISubsetS'$ must be matched to goods in $\GSubsetT$ so $\Abs{\ISubsetS'} \le \Abs{\GSubsetT}$.
\end{itemize}

So, we can conclude $\Abs{\ISubsetS'}=\Abs{\GSubsetT}$, therefore buyers in $\ISubsetS'$ and goods in
$\GSubsetT$ must be matched to each other at both equilibria. Furthermore, we can then conclude that buyers
in $\ISubsetS''$ and goods in $\GSubsetT''$ may only be matched to each other. That proves the claim of
\autoref{lem:entangle}. To complete the proof of \autoref{lem:match} observe that any good/buyer that has
positive price/utility at $\EQ$ either has a positive price/utility at $\EQ'$ in which case it must also be
matched at $\EQ'$, or has a $0$ price/utility at $\EQ'$ in which case it must be in $\GSubsetT$/$\ISubsetS$
and therefore it must also be matched at $\EQ'$.
\end{proof}


\begin{proof}[Proof of \autoref{thm:cont}]
We only prove the first claim. The proof of the second claim is similar (completely symmetric). The plan of
the proof is as follows:

First, we define a new market $\Market'=(\Buyers,\Goods, \UFun')$ with transformed utility functions
${\UFun'}_i^j(x) = \UFun_i^j(x+\PVecLB^j)-\UVecLB_i$. We claim that there is a one-to-one mapping between
$(\UVecLB,\PVecLB)$-bounded CEs of the original market and the CEs of the transformed market. Formally,
$\EQ=(\UVec,\PVec,\MVec)$ is a $(\UVecLB,\PVecLB)$-bounded CE of $\Market$ if and only if
$\EQ'=(\UVec-\UVecLB, \PVec-\PVecLB,\MVec)$ is a CE of $\Market'$. We then show that there is CE of
$\Market'$ in which there is a good with a price of $0$ which then means in the corresponding CE of the
original market the price of that good is equal to its lower bound and therefore at the lowest $(\UVecLB,
\PVecLB)$-bounded CE of $\Market$ the price of that good must also be equal to its lower bound which proves
the claim.

We now prove that there is a good with a price of $0$ at the lowest CE of $\Market'$. We choose an arbitrary
buyer $i$ from $\Market'$ and remove it from the market. Let $\EQH_{-i}$ be the highest CE of the remaining
market. By the assumption of the lemma, we know $\Abs{\Buyers}=\Abs{\Goods}$ and so in $\EQH_{-i}$ there are
more goods than there are buyers so there must be an unmatched good which we denote by $j^*$. Note that the
price of $j^*$ in $\EQH_{-i}$ must be $0$. On the other hand, by applying \autoref{thm:ind} to $\Market'$ and
using \eqref{eq:iii} we have $\PFun^j(\EQL) \le \PFun^j(\EQH_{-i})$ for every good $j$. Therefore, it must be
that the price of $j^*$ in $\EQL$ is also $0$ and that completes the proof.

There is a subtlety that we should point out about the one-to-one mapping between the CEs of the original
market and those of the transformed market. It is clear that every $(\UVecLB, \PVecLB)$-bounded CE of
$\Market$ can be transformed to a CE of $\Market'$. However, for the other direction, we need to show that
all goods/buyers are matched, otherwise after applying the inverse transform we may end up with an unmatched
good/buyer that has a positive price/utility. To show that all buyers/goods are matched in every CE of
$\Market'$, we can apply \autoref{lem:match}. To apply that lemma, we only need to show that there is a CE of
$\Market'$ in which all goods have strictly positive prices and then by that lemma all the goods must always
be matched (and so do all buyers because $\Abs{\Buyers}=\Abs{\Goods}$). Notice that if there is no CE for
$\Market'$ in which all goods have strictly positive prices then either in every $(\UVecLB, \PVecLB)$-bounded
CE of $\Market$ there is a good whose price is equal to its lower bound or $\Market$ has no $(\UVecLB,
\PVecLB)$-bounded CE at all which either way trivially proves the claim of this lemma.
\end{proof}



\begin{proof}[Proof of \autoref{thm:tight}] \
We only prove the first statement. The proof of the second statement is similar (completely symmetric).

First, we prove the ``only if'' direction. Assume that $\EQ$ is the lowest CE of $\Market$. For every subset
$\GSubsetT$ of goods with strictly positive prices, we prove that $\Abs{\DemandSet^\GSubsetT(\EQ)} \ge
\Abs{\GSubsetT}+1$, i.e. there are at least $\Abs{\GSubsetT}+1$ buyers who are interested in some good in
$\GSubsetT$. The proof is as follows. Since all the goods in $\GSubsetT$ have strictly positive prices, they
must all be matched. Let $\ISubsetS$ be the subset of buyers that are matched to $\GSubsetT$. Notice that
$\ISubsetS \subset \DemandSet^\GSubsetT(\EQ)$ and $\Abs{\ISubsetS}=\Abs{\GSubsetT}$. Therefore, to complete
the proof we only need to show that there is one more buyer not in $\ISubsetS$ who is also interested in a
good in $\GSubsetT$. Let $\PVecLB$ be the prices induced by utilities of buyers not in $\ISubsetS$, i.e.
$\PVecLB^j = \max_{i \in \Buyers-\ISubsetS} \PFun_i^j(\UFun_i(\EQ))$. Similarly, let $\UVecLB$ be the
utilities induced by the prices of goods not in $\GSubsetT$, i.e. $\UVecLB_i = \max_{j \in \Goods-\GSubsetT}
\UFun_i^j(\PFun^j(\EQ))$. Notice that $\EQ$ is a $(\UVecLB, \PVecLB)$-bounded CE of the market
$\Market'=(\ISubsetS,\GSubsetT, \UFun)$. Furthermore, if we replaced the part of $\EQ$ corresponding to
$\ISubsetS$ and $\GSubsetT$ with any other $(\UVecLB, \PVecLB)$-bounded CE of $\Market'$, we still get a
valid CE for $\Market$, but that implies $\EQ$ must already be the lowest $(\UVecLB, \PVecLB)$-bounded CE of
$\Market'$ as well because otherwise we could replace the part of $\EQ$ corresponding to $\ISubsetS$ and
$\GSubsetT$ with the lowest $(\UVecLB, \PVecLB)$-bounded CE of $\Market'$ and get a lower CE for $\Market$
which would contradict $\EQ$ being the lowest CE of $\Market$. By applying \autoref{thm:cont} to the market
$\Market'$, we can argue that there is a good $j^* \in \GSubsetT$ such that $\PFun^{j^*}(\EQ) =
\PVecLB^{j^*}$. Since all the goods in $\GSubsetT$, including $j^*$, have strictly positive prices,
$\PVecLB^{j^*}$ must also be strictly positive and because of the way we defined $\PVecLB$ there must be a
buyer $i^*$ not in set $\ISubsetS$ such that $\PVecLB^{j^*} = \PFun_{i^*}^{j^*}(\UVec_{i^*})$. That means
$i^*$ must be weakly interested in good $j^*$ and therefore $\{i^*\} \cup \ISubsetS \subset
\DemandSet^\GSubsetT(\EQ)$ which proves that $\DemandSet^\GSubsetT(\EQ) \ge \Abs{\GSubsetT}+1$.

The proof of the ``if'' direction is trivial. The proof is by contradiction. Let $\EQ$ be a CE of $\Market$
such that for every subset $\GSubsetT$ of goods with strictly positive prices we have
$\DemandSet^\GSubsetT(\EQ) \ge \Abs{\GSubsetT}+1$. Let $\EQL$ be the lowest CE of $\Market$ and assume that
$\EQ$ and $\EQL$ are not the same. Let $\GSubsetT$ consist of all the goods that have a higher price at $\EQ$
compared to $\EQL$. We know that there are at least $\Abs{\GSubsetT}+1$ buyers interested in $\GSubsetT$ at
$\EQ$ and these buyers must have higher utilities at $\EQL$ because the prices of the goods in $\GSubsetT$
are strictly lower. Therefore, the goods assigned to these buyers at $\EQL$ must have lower prices and so
there are at least $\Abs{\GSubsetT}+1$ goods that have higher prices at $\EQ$ compared to $\EQL$ which
contradicts the assumption that $\GSubsetT$ was the set of all the goods that had higher prices at $\EQ$.
\end{proof}


\begin{proof}[Proof of \autoref{thm:aa}]
The group strategyproofness follows from \autoref{thm:gsp}. So we only prove the second part. Suppose the
mechanism has computed a $\EQL=(\UVec,\PVec)$ which is the lowest CE of $\Market'$ as its outcome. Recall
that $\PVec^j$ is the virtual price of slot $j$ and the expected utility of advertiser $i$ from slot $j$ is
given by ${\UFun'}_i^j(\PVec^j)$. So for each advertiser $i\in \ISubsetS$ we have ${\UFun'}_i^j(\PVec^j) =
\CTR_i^j(\Val_i^j - \PVec^j/\CTRS_i^j)$. Furthermore, since $\CTR_i^j = \CTRS_i^j$, we can simplify the
utility function and get $\UFun_i^j(\PVec^j) = \CTR_i^j \Val_i^j - \PVec^j$. Now, consider the complete
bipartite graph with advertisers and slots. Let the weight of each edge $(i,j)$ be $\CTR_i^j \Val_i^j$. Note
that for each advertiser $i \in \ISubsetS$ we have $\UFun_i(\EQL)+\PFun^j(\EQL) \ge \CTR_i^j \Val_i^j$ in
which $\EQL$ is the outcome of the mechanism. Therefore, the total expected welfare of the coalition $\{s\}
\cup \ISubsetS$ is at least as much as the weight of the maximum weight matching in the absence of
$\ISubsetS'$. Furthermore, if $\ISubsetS'$ is empty (i.e. everyone agrees on the CTRs), the mechanism
computes the efficient allocation (i.e., a maximum weight matching) and the outcome is the same as the VCG
outcome.
\end{proof}

\end{document}